\documentclass[submission,copyright,creativecommons]{eptcs}
\providecommand{\event}{SR 2016} 
\usepackage{breakurl}             

\title{Representing Strategies}
\author{Hein Duijf \and Jan Broersen\thanks{Both authors gratefully acknowledge financial support from the ERC-2013-CoG project REINS, nr. 616512}}

\def\titlerunning{Representing Strategies}
\def\authorrunning{H.W.A. Duijf \& J.M. Broersen}

\usepackage[english]{babel}
\usepackage{amsmath}
\usepackage{amssymb}
\usepackage{amsthm}
\usepackage{enumitem}

\makeatletter

\newtheorem{definition}{Definition}
\newtheorem{proposition}{Proposition}
\newtheorem{corollary}{Corollary}
\newcommand{\thistheoremname}{}
\newtheorem*{genericthm}{\thistheoremname}

\def\sstit{\textup{\;\text{sstit}}}
\def\acc{\textup{\;\text{acc}}}
\def\Ags{\textup{\textit{Ags}}}

\def\Prop{\textit{Prop}}

\def\STIT{\textup{\ensuremath{\mathsf{STIT}}}}
\def\GSTRAT{\textup{\ensuremath{\mathsf{G.STRAT}}}}
\def\ATL{\textup{\ensuremath{\mathsf{ATL}}}}
\def\ATEL{\textup{\ensuremath{\mathsf{ATEL}}}}

\let\textquotedbl="

\def\lnext{\textup{\ensuremath{\mathsf{X}}}}
\def\henceforth{\textup{\ensuremath{\mathsf{G}}}}

\def\Ag{\textup{\ensuremath{\mathsf{Ag}}}}
\def\St{\textup{\ensuremath{\mathsf{St}}}}
\def\Act{\textup{\ensuremath{\mathsf{Act}}}}
\def\act{\textup{\ensuremath{\mathsf{act}}}}
\def\out{\textup{\ensuremath{\mathsf{out}}}}
\def\CGS{\ensuremath{(\Ag,$ $\St,$ $\Act,$ $\act,$ $\out)}}
\def\CGM{\ensuremath{(\Ag,$ $\St,$ $\Act,$ $\act,$ $\out,$ $L)}}
\def\CEGM{\ensuremath{(\Ag,$ $\St,$ $\{\sim_i\mid i\in\Ag\},$ $\Act,$ $\act,$ $\out,$ $L)}}
\def\outset{\ensuremath{\out\texttt{\char`_}\textup{\ensuremath{\mathsf{set}}}}}
\def\outplays{\ensuremath{\out\texttt{\char`_}\textup{\ensuremath{\mathsf{plays}}}}}

\def\Cond{\textup{\ensuremath{\mathsf{Cond}}}}
\def\Eff{\textup{\ensuremath{\mathsf{Eff}}}}
\def\perf{\textup{\;\text{perf}}}

\usepackage{comment}
\newif\ifshow
\showtrue
\ifshow
  
\else
  \excludecomment{tbd}
\fi

\makeatother

\begin{document}

\maketitle

\begin{abstract}
Quite some work in the ATL-tradition uses the differences between various types of strategies (positional, uniform, perfect recall) to give alternative semantics to the same logical language. This paper contributes to another perspective on strategy types, one where we characterise the differences between them on the syntactic (object language) level. This is important for a more traditional knowledge representation view on strategic content. Leaving differences between strategy types implicit in the semantics is a sensible idea if the goal is to use the strategic formalism for model checking. But, for traditional knowledge representation in terms of object language level formulas, we need to extent the language. This paper introduces a strategic STIT syntax with explicit operators for knowledge that allows us to charaterise strategy types. This more expressive strategic language is interpreted on standard ATL-type concurrent epistemic game structures. We introduce rule-based strategies in our language and fruitfully apply them to the representation and characterisation of positional and uniform strategies. Our representations highlight crucial conditions to be met for strategy types. We demonstrate the usefulness of our work by showing that it leads to a critical reexamination of coalitional uniform strategies.
\end{abstract}

\section{Introduction}


To make way for strategic reasoning on the syntactic level we need to know how to represent the various types of strategies that have been proposed in the literature. In this paper, our aim in particular will be to provide syntactic counterparts for the various strategy types proposed in the $\ATL$-tradition, such as positional strategies \cite{alur_alternating-time_2002} and uniform strategies \cite{jamroga_agents_2004}. Our proposed language, an extension of strategic $\STIT$ including temporal and epistemic modalities and action types, is sufficiently expressive for representing positional and uniform strategies.%
\footnote{Recently, a syntactic characterization of uniform strategies in \emph{Epistemic Strategy Logic} was presented \cite{knight_dealing_2015}. Our representation differs in that we use rule-based strategies.}
Whereas the introduction of positional and uniform strategies are intuitively appealing, our characterizations elucidate the underlying conditions to be met for these semantically defined strategy types. Although emulating ability-modalities akin to the $\ATL$-tradition is not the main objective here, some hints are provided that guide such a future endeavour. 

Our representation of uniform positional strategies enhances the understanding of a \emph{coalition's} uniform strategy. Our result naturally invites different ways to distribute relevant strategic knowledge. This observation complements \cite{agotnes_knowledge_2015} where it is argued that a coalition's uniform strategy is imprecise only in the ``mode'' of the coalition's knowledge, referring either to common, distributive or mutual knowledge. 

To further explain our view and approach, we will subsequently answer the following three questions:

\begin{enumerate}
\item what \emph{is} a strategy?
\item how to represent the \emph{performance} of a strategy?
\item how to \emph{characterize} a rule-based strategy?
\end{enumerate}

{\bf what is a strategy?}\\
It is remarkable how strong the notion of a strategy varies throughout the literature on strategic reasoning. In $\ATL$ frameworks\footnote{See the seminal work \cite{alur_alternating-time_2002} and extensions such as $\mathsf{CATL}$ \cite{van_der_hoek_logic_2005}, and \textit{Strategy Logic} \cite{chatterjee_strategy_2007} and \cite{mogavero_reasoning_2013}.} a strategy is a mapping that assigns an action type to every finite history of system states; in $\STIT$ frameworks\footnote{See the seminal work \cite{belnap_facing_2001}, \cite{horty_agency_2001}  and the recent extension to strategic action \cite{broersen_stit-logic_2009}.}, since acting means restricting the possible futures, a strategy is identified with the futures it allows; finally Dynamic Logics\footnote{See \cite{van_benthem_extensive_2002}. 
} take strategies to be temporally extended act structures of sequences of (atomic) actions. This paper starts with the $\ATL$ conception of strategy types and ultimately provides representations thereof. 

{\bf how to represent the \emph{performance} of a strategy?}\\
To express that a strategy is actually performed, we use insights from $\STIT$ frameworks. Other frameworks express strategy performance only implicitly, safely tucked away under path quantifiers ($\ATL$) or in quantifiers in the central modalities (Dynamic logic). This follows from the observation that he main operators $\langle\langle C\rangle\rangle\varphi$ and $[\alpha]\varphi$ in these systems are interpreted as `coalition $C$ is able to ensure $\varphi$' and `after executing action $\alpha$, $\varphi$ holds'. This reveals that $\ATL$ enables one to reason about strategic ability and Dynamic Logics support reasoning about the results of actions, but not reasoning about the \emph{performance} of actions or strategies here and now.

{\bf how to \emph{characterize} a rule-based strategy?}\\
Strategies are typically communicated in the form of condition-action rules. Therefore it makes sense to also logically represent them in that form using a suitable language. Recently, $\ATL$ has been extended to enable reasoning about rule-based strategies in \cite{zhang_representing_2014}. One of the main themes there is the representation of (semantic) strategies by formulas of their proposed language, which includes rule-based strategies. A different study in \cite{ramanujam_dynamic_2008} evaluates formulas at game-strategy pairs thereby combining aspects of \emph{Game Logic}\footnote{See the original work \cite{parikh_logic_1985} and the overview \cite{pauly_game_2003}.} and strategic reasoning. They propose a multi-sorted language to express the \emph{structure} of the strategy. The logic includes two types of conditional strategies $[\psi\mapsto a]^i$ and $\pi\mapsto\sigma$, which are interpreted as `player $i$ chooses move $a$ whenever $\psi$ holds' and `player $i$ sticks to the specification given by $\sigma$ if on the history of play, all moves made by $\bar{i}$ conform to $\pi$', respectively. The $\psi$ in the first formula is restricted to boolean combinations of propositional letters, and the second formula represents that the other players \emph{have acted} in accordance with $\pi$. We, however, do not want to commit ourselves to these restrictions in the language or regarding the conditions; Most importantly, to represent uniform strategies epistemic conditions have to be allowed.

The paper is organized as follows: In Section \ref{Section: CGM} the well-known Concurrent Epistemic Game Models (see \cite{alur_alternating-time_2002} and \cite{van_der_hoek_cooperation_2003}) are introduced to provide the basis for our semantics. In contrast to the usual $\ATL^{\ast}$-divide between path and state formulae, a key idea from $\GSTRAT$%
\footnote{$\GSTRAT$ is first introduced in \cite{broersen_stit-logic_2009}. It is an extension of basic $\STIT$ frameworks to a strategic and multi-agent setting.}
to evaluate formulas against tuples consisting of the state, path and \emph{the current strategy profile} is implemented in Section \ref{Section: Logical approach}. In parallel, a logical language is introduced including action types and temporal, epistemic and agency operators. As is common in $\STIT$ frameworks, the structure of a strategy can be described by specifying the (temporal) properties it ensures. This allows us to view a rule-based strategy as a set of condition-effect rules in Section \ref{Section: Rule based strategies}, which are then used in Section \ref{Section: Respresentation Results} to investigate how this logical language can be fruitfully applied to represent strategy types. Certain types of rule-based strategies are used to represent positional and uniform strategies, thereby uncovering crucial conditions for these semantic strategy types. Some novel implications are drawn on coalition's uniform strategies by appealing, not to different ``modes'' of coalitional knowledge but, to ways of distributing the relevant strategic knowledge. We conclude in Section \ref{Conlusion}, proof sketches are to be found in the appendix. 


\section{\label{Section: CGM}Concurrent Game Models}

In this section we introduce Concurrent Game Models (see \cite{alur_alternating-time_2002}) and Concurrent Epistemic Game Models (see \cite{van_der_hoek_cooperation_2003}). Our treatment will be roughly in line with \cite{agotnes_knowledge_2015} although our notation on histories will differ to neatly support the syntactic approach in the next section.

\begin{definition}[Concurrent game structure and model]\label{Definition: CGM} 
A \emph{concurrent game structure }(CGS) is a tuple $\mathcal{S}=\CGS$ which consists of:

\begin{enumerate}
\item a finite, non-empty set of \emph{agents} $\Ag=\{1,$ $\ldots,$ $K\}$;
the subsets of $\Ag$ are called \emph{coalitions};
\item a non-empty set of \emph{states} $\St$;
\item a non-empty set of \emph{action types} $\Act$;
\item an \emph{action manager }function $\act:\Ag\times\St\rightarrow\mathcal{P}(\Act)$ assigning to every player $i$ and state $q$ a non-empty set of actions
available for execution by $i$ at $q$. \\
An \emph{action profile} is a tuple of actions $\alpha=\langle\alpha_{1},\ldots,\alpha_{K}\rangle\in\Act^{K}$. The action profile is \emph{executable at the state} $q$ if $\alpha_{i}\in\act(i,q)$ for every $i\in\Ag$. We denote by $\act(q)$ the subset of $\prod_{i\in\Ag}\act(i,q)$ consisting of all action profiles executable at the state $q$.
\item 
\begin{itemize}
\item[(a)] a \emph{transition function} $\out$ that assigns an \emph{outcome state} $\out(q,\alpha)$ to every state $q$ and every action profile $\alpha\in\act(q)$.%
\footnote{In \cite[p.~553 -- our boldfacing, emphasis in original]{agotnes_knowledge_2015} they introduce $\out$ as ``a \emph{transition function} $\out$ that assigns a \textbf{unique} \emph{outcome state} $\out(q,\alpha)$ to every state $q$ and every action profile $\alpha$ which is executable at $q$''.  We are puzzled about what this required uniqueness adds to the fact that $\out$ is a function. Although they do not explicitly mention what is meant, we think this uniqueness is represented in (b). The examples that they consider are in line with (b). This addition is crucial for our main Proposition \ref{Proposition: main result}(3) as it implies that two strategies are play-equivalent if and only if they are identical.} 
\item[(b)] for every two states $q,q'$ there is at most one $\alpha\in\act(q)$ such that $\out(q,\alpha)=q'$, i.e.~if a state transition is labelled, then it has a unique label.
\end{itemize}
\end{enumerate}

A \emph{concurrent game model }(CGM) is a CGS endowed with a labeling $L:\St\rightarrow\mathcal{P}(\Prop)$ of the states with sets of atomic propositions from a fixed set $\Prop$. As usual, the labeling describes which atomic propositions are true at a given state.
\end{definition}

We fix a concurrent game structure $\mathcal{S}=\CGS$. An action profile is used to determine a successor of a state using the transition function $\out$. The set of the available action profiles is denoted by $\act(q)$, consequently the set of \emph{possible successors} of $q$ is the set of states $\out(q,\alpha)$ where $\alpha$ ranges over $\act(q)$. An infinite sequence $\lambda=q_0 q_1 q_2 \cdots$ of states from $\St$ is called a \emph{play} if $q_{k+1}$ is a successor of $q_k$ for all positions $k\geq 0$. $\lambda[k]$ denotes the $k$-th component $q_k$ in $\lambda$, and $\lambda[0,k]$ denotes the initial sequence, or \emph{history}, $q_0\cdots q_k$ of $\lambda$. 

A \emph{perfect recall strategy for an agent} $i$ is a function $s_i$ that maps every history $\lambda[0,k]$ to an action type $s_i(\lambda[0,k])\in\act(\lambda[k])$. A \emph{positional (aka memoryless) strategy for an agent} $i$ is a function $s_i$ that maps every state $q$ to an action type. A \emph{perfect recall strategy for a coalition} $C\subseteq\Ags$, also called a coalitional strategy for $C$, is a function $s_C$ mapping each agent $i\in C$ to a perfect recall strategy $s_C(i)$. Positional strategies for coalitions are defined analogously. A \emph{strategy profile} $s$ is a coalitional strategy for $\Ag$. A coalitional strategy $s_C$ \emph{extends} $s'_{C'}$, notation $s_C\sqsupseteq s'_{C'}$, if and only if $C\supseteq C'$ and $s'_{C'}(i)=s_C(i)$ for every $i\in C$. Given a strategy profile, we often write $s_C$ for the coalitional strategy for $C$ satisfying $s_C\sqsubseteq s$.

The set $\outplays(\lambda[0,k],s_C)$ of outcome plays of a strategy $s_C$ for $C$ at a history $\lambda[0,k]$ is the set of all plays $\lambda'=q_0q_1\cdots$ such that $\lambda'[0,k]=\lambda[0,k]$ and, for every $l\geq k$, there is an action profile $\alpha=\langle\alpha_1,\ldots,\alpha_K\rangle$ satisfying $\alpha_i=s_C(i)(\lambda'[l])$ for all $i\in C$ and $q_{l+1}=\out(q_l,\alpha)$.

Our formal results rely on the notion of \emph{play-equivalence}:

\begin{definition}[Play-equivalence] Let $\lambda$ be a play, let $k$ be a position, let $C$ a coalition, and let $s_C$, $s'_C$ perfect recall strategies. We say that $s_C$ and $s'_C$ are \emph{play-equivalent} at $\lambda[0,k]$ if and only if\\ $\outplays(\lambda[0,k],s_C)=\outplays(\lambda[0,k],s'_C)$.
\end{definition}

It is standard to model the agent's incomplete information by extending concurrent game models:

\begin{definition}[Concurrent Epistemic Game Model] \label{Definition: CEGM} 
A \emph{concurrent epistemic game model }(CEGM) is a tuple
$\mathcal{M}=\CEGM$
which consists of a CGM $\CGM$ and indistinguishability
relations $\sim_{i}\subseteq\St\times\St$, one for each agent. Furthermore, it is assumed that (1) $\sim_{i}$ is an equivalence relation, and (2) $q\sim_{i}q'$ implies that $\act(q,i)=\act(q',i)$.
\end{definition}

These indinstinguishability relations are straightforwardly extended to histories by: $\lambda[0,k]\sim_{i}\lambda'[0,k']$  iff $k=k'$ and for every $l\leq k$ we have $\lambda[l]\sim_{i}\lambda'[l]$. Then we introduce a third strategy type: 
\begin{itemize}
\item[] a \emph{uniform strategy for an agent} $i$ is a perfect recall strategy $s_i$ satisfying: for all histories $\lambda[0,k]$, $\lambda'[0,k']\in\St^{+}$, if $\lambda[0,k]\sim_{i}\lambda'[0,k']$ then $s_{i}(\lambda[0,k])=s_{i}(\lambda'[0,k'])$.
\end{itemize}

A \emph{uniform coalitional strategy for coalition} $C$ is a function $s_C$ mapping each agent $i\in C$ to a uniform strategy $s_C(i)$.


In the remainder, we mean ``perfect recall strategies" when writing ``strategies'', unless otherwise specified.

\section{\label{Section: Logical approach}Strategic language}

In the previous section we outlined the models that provide the basis for the semantics of our logical enterprise. In the current section we introduce our logical framework, which is inspired by \cite{broersen_stit-logic_2009}.%
\footnote{Whereas \cite{broersen_stit-logic_2009} uses models based on the $\STIT$-tradition, here Concurrent Game Models are used for interpreting the language.} 
First, we introduce the syntax of the logical language. Second, we present the truth conditions of the logical formalism. It is crucial that we evaluate formulas with respect to tuples $\langle k,\lambda,s\rangle$ which include \emph{the current strategic course of action} $s$ (inspiration from \cite{broersen_stit-logic_2009}). Finally, some crucial observations on the resulting logical formalism are presented by reviewing the underlying models. 

\begin{definition}[Syntax]\label{Definition: Syntax} 
Fix a set of propositional letters $\Prop$, a finite set of agents $\Ags$, and a set of action types $\Sigma$. The formulas of the language $\mathcal{L}$ are given by: 

\begin{center}
$\varphi::=p\mid\alpha_{C}\mid\varphi\land\varphi\mid\lnot\varphi\mid\lnext\varphi\mid\henceforth\varphi\mid\Box\varphi\mid[C\sstit]\varphi\mid K_{i}\varphi$,
\end{center}

\noindent where $p$ ranges over $\Prop$, $C$ ranges over subsets of $\Ags$,
and $\alpha_{C}$ ranges over $\Sigma^{C}$. 
\end{definition}

Given a CEGM $\mathcal{M}=\CEGM$ with $\Ag=\Ags$ and $\Act=\Sigma$, these formulas will be evaluated at tuples $\langle k,\lambda,s\rangle$  consisting of a strategy profile $s$, a play $\lambda$ such that $\lambda\in\outplays(\lambda[0],s)$, and a position $k$. This means that the truth of formulas is evaluated with respect to a current state $\lambda[k]$, a current history $\lambda[0,k]$, a current future $\lambda[k,\infty)$, and a current strategy profile $s$. Obviously, by incorporating the current strategy profile into the worlds of evaluation we get a semantic explication of the performance of a strategy.

The central agency operator is the modality $[C\sstit]\varphi$ which stands for `the coalition $C$ strategically sees to it that $\varphi$ holds'. Relative to a tuple $\langle k,\lambda,s\rangle$ the modality $[C\sstit]\varphi$ is interpreted as `the coalition $C$ is in the process of executing $s_{C}$ thereby ensuring the (temporal) condition $\varphi$'. In addition, the language includes temporal modalities $\lnext\varphi$ and $\henceforth\varphi$ which are interpreted, relative to a tuple $\langle k,\lambda,s\rangle$, as `$\varphi$ holds in the next moment after $\lambda[k]$ on $\lambda$' and `$\varphi$ holds on all future moments after $\lambda[k]$ on $\lambda$', respectively. In contrast to this longitudinal dimension of time, the language includes a temporal modality $\Box\varphi$ for historical necessity. The modality $\Box\varphi$ is interpreted, relative to a tuple $\langle k,\lambda,s\rangle$, as `$\varphi$ holds on any tuple at $\lambda[0,k]$'. This highlights that the truth of $\Box\varphi$ does not depend on the dynamic aspects represented by the current future and the current strategy profile, we call such formulas \emph{moment-determinate}. Finally, we include epistemic modalities $K_{i}\varphi$, one for each agent, which are interpreted as `agent $i$ knows that $\varphi$'. The presented syntax and semantics are formally connected by the truth conditions for the syntactic clauses:

\begin{definition}[Semantics]
\label{Definition: Semantics} Let $\mathcal{M}=\CEGM$
be a CEGM with $\Ag=\Ags$ and $\Act=\Sigma$. The points of evaluation
for our logical formulas are tuples $\langle k,\lambda,s\rangle$
consisting of a strategy profile $s$, a play $\lambda\in\outplays(\lambda[0],s)$,
and a position $k$. The truth conditions are given by inductive definitions
(suppressing the model $\mathcal{M}$ and not listing standard propositional truth conditions):

\begin{description}
\item $\langle k,\lambda,s\rangle\vDash\alpha_{C}$ $\Leftrightarrow$ $s_{C}(\lambda[0,k])=\alpha_{C}$
\item $\langle k,\lambda,s\rangle\vDash\lnext\varphi$ $\Leftrightarrow$ $\langle k+1,\lambda,s\rangle\vDash\varphi$ 
\item $\langle k,\lambda,s\rangle\vDash\henceforth\varphi$ $\Leftrightarrow$ for each $l\geq k$  : $\langle l,\lambda,s\rangle\vDash\varphi$ 
\item $\langle k,\lambda,s\rangle\vDash\Box\varphi$ $\Leftrightarrow$ for every strategy $s'$ and every $\lambda'\in\outplays(\lambda[0,k],s')$ we have $\langle k,\lambda',s'\rangle\vDash\varphi$
\item $\langle k,\lambda,s\rangle\vDash[C\sstit]\varphi$ $\Leftrightarrow$ for every $s'\sqsupseteq s_{C}$ and every $\lambda'\in\outplays(\lambda[0,k],s')$ we have $\langle k,\lambda',s'\rangle\vDash\varphi$ 
\item $\langle k,\lambda,s\rangle\vDash K_{i}\varphi$ $\Leftrightarrow$ for every $\lambda', k'$  such that $\lambda'[0,k']\sim_{i}\lambda[0,k]$ and any $s'$ we have $\langle k',\lambda',s'\rangle\vDash\varphi$.\footnote{Little is known about properties of this logic and such a formal inquiry would lead us to far astray from the current enterprise. To guide some of the formal intuitions of the reader we mention some validities without proof or conceptual motivation: 
$\Box$ and $[i\sstit]$ are $\mathsf{S5}$-modalities, 
the $[i\sstit]$-operator is monotone in its agency argument, i.e.\,for $A\subseteq B$ we have $\vDash[A\sstit]p\rightarrow[B\sstit]p$,
the temporal part is the standard discrete linear temporal logic containing $\lnext$ and $\henceforth$,
some interaction principles are $\vDash\Box\lnext p\rightarrow \lnext\Box p $, 
$\vDash[C\sstit]\lnext p\rightarrow \lnext[C\sstit] p $, 
$\vDash[C\sstit]\Box p \leftrightarrow \Box p$,
$\vDash[C\sstit]\lnext[C\sstit]p \leftrightarrow [C\sstit]\lnext p$,
$\vDash[C\sstit]\henceforth[C\sstit]p \leftrightarrow [C\sstit]\henceforth p$,
and $\Diamond[A\sstit]p\land\Diamond[B\sstit]q\rightarrow\Diamond[A\cup B\sstit](p\land q)$ for disjoint coalitions $A$ and $B$ (independence of agency). }
\end{description}
\end{definition}

With the semantics in place, we gather some crucial observations:

\begin{itemize}

\item Because the truth of a propositional letter only depends on the current state, it is not surprising that the truth of any propositional formula only depends on the current state. Therefore, we will often write $\lambda[k]\vDash\varphi$ instead of $\langle k,\lambda,s\rangle\vDash\varphi$ for a propositional formula $\varphi$. This is connected to the familiar divide in $\ATL^\ast$ syntax between state and path formulas. 
\item As mentioned before, a formula $\varphi$ is moment-determinate if $\vDash\varphi\leftrightarrow\Box\varphi$. The truth of such formulas only depends on the history, so we will often write $\lambda[0,k]\vDash\varphi$ instead of $\langle k,\lambda,s\rangle\vDash\varphi$ for such formulas. 
\item Formulas of the form $K_i\varphi$ are moment-determinate. In particular, an agent does not know what he is doing.
\item  The truth of a coalitional action type only depends on the current strategy of that coalition, i.e. $\langle k,\lambda,s\rangle\vDash\alpha_C\leftrightarrow[C\sstit]\alpha_C$.
\item The formula $\Diamond\alpha_C$ expresses that a coalitional action type $\alpha_C$ is executable at a state. 
\item Observe that only the truth conditions for the $[C\sstit]$-operator and the action types $\alpha_C$ involve the current $C$-strategy profile. It is clear that adding (at least one of) these is necessary to express that a certain strategy is performed. The action types are inherited from a bottom-up perspective on strategies with the action types as atomic building blocks. In contrast, the $[C\sstit]$-operator incorporates a top-down view in that a strategy is described by the properties it ensures. \end{itemize}

\section{\label{Section: Rule based strategies}Rule-based strategies}

A rule-based strategy consists of rules. Such a rule is composed of a condition and an effect, thereby incorporating the intuition that a rule is triggered under certain conditions and has a certain effect: 

\begin{definition}[Rule-based strategies]
\label{Definition: Rule-based strategies} A \emph{rule-based strategy} $RS$ is a finite set of condition-effect rules $\{c_{1}\mapsto e_{1},\ldots,c_{N}\mapsto e_{N}\}$, which represents that condition $c_{n}$ triggers effect $e_{n}$. We denote the set of conditions occurring in such a rule-based strategy $RS$ by $\Cond(RS)$, likewise the effects by $\Eff(RS)$. 
\end{definition}

Performing a rule-based strategy means that in case a rule is triggered one ensures that the corresponding effect is realized:

\begin{definition}[Performing rule-based strategies]
\label{Definition: Performing rule-based strategies} Henceforth we fix a CEGM $\CEGM$ with $\Ag=\Ags$ and $\Act=\Sigma$. Let $RS=\{c_{1}\mapsto e_{1},\ldots,c_{N}\mapsto e_{N}\}$
be a rule-based strategy. First, we say that \emph{a coalition
acts according to rule-based strategy }$RS$ \emph{at} $\langle k,\lambda,s\rangle$, denoted by $\langle k,\lambda,s\rangle\vDash[C\acc]RS$,
if and only if 
\[
\langle k,\lambda,s\rangle\vDash[C\sstit]\bigwedge_{n\leq N}(c_{n}\rightarrow[C\sstit]e_{n}).
\]


\noindent Second, we say that \emph{a coalition $C$ performs }$RS$\emph{ at}
$\langle k,\lambda,s\rangle$, denoted by $\langle k,\lambda,s\rangle\vDash[C\perf]RS$,
if and only if 
\[
\langle k,\lambda,s\rangle\vDash[C\sstit]\henceforth[C\acc]RS.
\]
\end{definition}

The formula $[C\acc]S$ is interpreted, relative to a tuple $\langle k,\lambda,s\rangle$, as `coalition $C$ is in the process of executing strategy $s_C$ thereby ensuring that the conditionals are met'. Informally, it means that coalition $C$ is currently performing a strategy that ensures that in case a condition holds he performs a strategy ensuring the corresponding effect.

Although the nested $[C\sstit]$ operator may be puzzling at first sight, it makes perfect sense. To argue in favour we break the formula down. A rule of the rule-based strategy is formalized as $c\rightarrow [C\sstit]e$, but one should not forget that here we intend to formulate that a coalition is acting according to such a rule-based strategy. This is expressed by the second $[C\sstit]$ operator, which guarantees that one is acting accordingly not only at the current play, but also at all plays in $\outplays(\lambda[0,k],s_C)$. To formalize that a coalition is performing a rule-based strategy, we add the $\henceforth$ operator to express that it is henceforth acting according to strategy $RS$.

There are two ways in which a rule-based strategy can be unsatisfactory: (a) the agent might not be able to perform a certain rule-based strategy, or (b) the action description given by a certain rule-based strategy can be underspecified. So a rule-based strategy can be viewed as a partial perfect recall strategy which is defined at a history if and only if it is possible to act accordingly and there is but one way to do so.\footnote{This resembles the informal notion of deterministic strategies in \cite[p.~204]{zhang_representing_2014}: ``move recommendations are always unique'' for deterministic strategies.} To investigate this more thoroughly, we continue in our logical framework.

Can a given rule-based strategy $RS$ be viewed as a partial strategy? There is a straightforward way to attempt this whenever the conditions are moment-determinate
, i.e.~$\mathcal{M}\vDash c\leftrightarrow\Box c$ for each $c\in\Cond(RS)$.%
\footnote{If the conditions are not moment-determinate ``a choice of an agent, at a given point of a play, may depend on choices other agents can make in the future or in counterfactual plays" (cf.~the study on "behavioral strategies" in \cite[p.~149]{mogavero_behavioral_2014}). } 
Using $RS$, we define a \emph{partial} coalition strategy $s^{RS}_{C}:\St^{+}\rightarrow\Act^{C}$ by:

\begin{center}
\begin{tabular}{lll}
$s^{RS}_{C}(\lambda[0,k])=\alpha_{C}$  & iff &  $\lambda[0,k]\vDash\Diamond[C\acc]RS$ $\land $ $\Box([C\acc]RS\rightarrow[C\sstit]\alpha_C)$.
\end{tabular}
\end{center}

The first conjunct says that the coalition is able to act accordingly, whereas the second conjunct says that performing action profile $\alpha_C$ is the only way to do so. 

Clearly, this partial coalition strategy is defined at a history $\lambda[0,k]$ if the following conditions hold:
\begin{description}
\item [{(1)}] there is a $c\in\Cond(RS)$ such that $\lambda[0,k]\vDash c$,

\item [{(2)}] there is at least one $\alpha_{C}$ such that $\lambda[0,k]$ $\vDash$ $\Diamond[C\sstit]\alpha_C$ $\land$ $\Box([C\sstit]\alpha_C\rightarrow[C\acc]RS)$, and

\item [{(3)}] there is at most one $\alpha_{C}$ such that $\lambda[0,k]$ $\vDash$ $\Diamond[C\sstit]\alpha_C$ $\land$ $\Box([C\sstit]\alpha_C\rightarrow[C\acc]RS)$.

\end{description}

The failure of (1) and the failure of (3) signify that the rule-based strategy is underspecified either because no rule has been triggered or because there are multiple ways to act accordingly. The failure of (2), however, indicates a practical inconsistency or a conflict in the rule-based strategy $RS$, because it implies that there is no way to act accordingly.

Before proceeding, we extend $\outplays(-,-)$ to pertain also to partial strategies:

\begin{definition}
Let $s_C$ be a partial perfect recall strategy. We define 
\begin{center}
 $\outplays(\lambda[0,k],s_{C})$  $:=$  $\{\lambda'\in\St^{\omega} \mid\lambda'\sqsupset\lambda[0,k]$ and for every $l\geq k$ we have: if $s_{C}(\lambda'[0,l])$ is defined then $\lambda[l+1]\in\outset(\lambda[0,l],s_{C})\}$.
\end{center}
\end{definition}

To investigate the perfect recall strategies represented by a rule-based strategy, we use the partial strategy it defines:

\begin{proposition}
\label{Proposition: main result} Let $RS$ be a rule-based
strategy with only moment-determinate conditions. Let $s^{RS}_{C}$ be
the partial coalition strategy defined by $RS$. 
Then 
\begin{enumerate}
\item for any profile $\langle k',\lambda',s'\rangle$ such that $s^{RS}_C$ is defined on $\lambda'[0,k']$ we have that the following are equivalent: (a) $\langle k',\lambda',s'\rangle\vDash[C\acc]RS$ and (b) $s^{RS}_C(\lambda'[0,k'])=s'_C(\lambda'[0,k'])$.
\item Let $\langle k',\lambda',s'\rangle$ be a profile such that $s^{RS}_{C}$ is defined on all histories in $\outplays(\lambda'[0,k'],s'_{C})$. Then $\langle k',\lambda',s'\rangle\vDash[C\perf]RS$ if and only if $s^{RS}_{C}$ and $s_{C}'$ are play-equivalent at $\lambda'[0,k']$.
\end{enumerate}
\end{proposition}


This establishes a crucial connection between the syntactic notion of performing a rule-based strategy and the semantic notion of a (partial) perfect recall strategy. In the following subsections we use this link to represent positional and uniform strategies up to play-equivalence by rule-based strategies. 

The notion of play-equivalence stems from the $\STIT$ views in our formalism. From a $\STIT$ perspective a strategy is identified by the futures it allows, so two play-equivalent strategies not only \emph{appear} to be same strategy, they \emph{are} the same strategy. 

\section{Representation Results}\label{Section: Respresentation Results}

\subsection{Representing positional strategies}\label{Subsection: Representing positional strategies}

In this subsection, we prove that rule-based strategies can be used to represent positional strategies. For that purpose we introduce a specific type of rule-based strategies:

\begin{definition}
\label{Definition: Proposition-action strategies} 
A \emph{proposition-action strategy} $RS$ for a coalition $C$ is
a rule-based strategy such that the conditions $c$ and effects $e$
are respectively of the form:
\begin{center}
\begin{tabular}{lll}
$c  ::=  p\mid c\land c\mid\lnot c$ & \quad & $e  ::=  \alpha_{C}$,
\end{tabular}
\end{center}
where $p$ ranges over $\Prop$ and $\alpha_{C}$ ranges over $\Sigma^{C}$. 
\end{definition}

Because the effects are of the form $\alpha_C$, there can be at most one way to act according to a proposition-action strategy whenever one of the conditions is triggered.  This motivates our definition of \emph{completeness}; a notion that plays a key role in our findings in the correspondence between proposition-action strategies and positional strategies:

\begin{definition}[Completeness]
\label{Definition: Completeness of strategies}
We say that \emph{a rule-based strategy $RS=\{c_{1}\mapsto e_{1},\ldots,c_{N}\mapsto e_{N}\}$
is complete at} $\lambda[0,k]$ if and only if $\lambda[0,k]\vDash\Box\henceforth\bigvee_{n\leq N}c_{n}$,
i.e.~in any possible future point one of the conditions in $S$ is
triggered.\footnote{Compare \cite[p.~207]{zhang_representing_2014}:  ``a complete strategy provides the player with a ``complete guideline'' that always provides the player with one or more suggestions how to act when it is his move''. }
\end{definition}

\begin{proposition}
\label{Proposition: Representing positional strategies}Let $s$ be a strategy profile, and let $\lambda\in\outplays(\lambda[0],s)$. Suppose there is a proposition-action strategy $RS$ for $C$ that is complete at $\lambda[0]$ such that $\langle0,\lambda,s\rangle\vDash[C\perf]RS$.
Then there is a positional strategy $\hat{s}_{C}$ that is play-equivalent to $s_{C}$ at $\lambda[0]$. 
\end{proposition}


This shows that performing a complete proposition-action strategy implies that one is performing a strategy that is play-equivalent to a positional strategy. In a sense, this means that the strategies represented by complete proposition-action strategies are positional strategies. 

The converse does not hold in general, which can be shown by providing a CGM containing a positional strategy differing at two propositionally equivalent states. So to prove the converse we have to restrict our investigation to CGMs in which enough states are propositionally definable:

\begin{proposition}
\label{Proposition: Propositional definability and representing positional strategies}
Let $s_{C}$ be a positional $C$-strategy,and let $\lambda\in\outplays(\lambda[0],s_{C})$. Suppose that $s_{C}(\St)\subseteq\Act^{C}$ is finite and that every $s_{C}^{-1}(\alpha_{C})\subseteq\St$ is propositionally definable. Then there is a proposition-action strategy $RS$ for $C$ that is complete at $\lambda[0]$ such that $\langle0,\lambda,s\rangle\vDash[C\perf]RS$.

\end{proposition}

This shows that, under certain semantic constraints, a given positional strategy is represented by a complete proposition-action strategy. Thereby we can move strategic reasoning in the semantics about positional strategies to reasoning on the syntactic level about proposition-action strategies. In conclusion, we show that in a common class of CGMs, complete proposition-action strategies correspond to positional strategies up to play-equivalence: 

\begin{corollary}
\label{Corollary: Characterizing positional strategies}
Let $s$ be a perfect recall strategy for $C$, and let $\lambda\in\outplays(\lambda[0],s_{C})$. Let $\St$ be finite and let every state be propositionally definable.
Then the following are equivalent
\begin{enumerate}
\item there is a proposition-action strategy $RS$ for $C$ such that $\langle k,\lambda,s\rangle\vDash[C\perf]RS\land\Box\henceforth\bigvee_{c\in\Cond(RS)} c$,
\item at $\lambda[0]$ the perfect recall strategy $s$ is play-equivalent
to a positional strategy.
\end{enumerate}
\end{corollary}

This corollary uncovers that completeness is an underlying condition  for positional strategies. Although this discovery is unsurprising and intuitive, it shows that our language is able to express such underlying intuitions. 


\subsection{Representing uniform strategies}\label{Subsection: Representing uniform strategies}

Here we prove that rule-based strategies are useful for representing uniform strategies, focussing on individuals' uniform strategies, using a type of rule-based strategies: 

\begin{definition}
\label{Definition: Knowledge-action strategies}
A \emph{knowledge-action strategy $RS$ }for a coalition $C$ is a
rule-based strategy such that the conditions are of the
form $K_{i}c$ with $c$ and the effects $e$ respectively of the form:

\begin{center}
\begin{tabular}{lll}
$c ::= p\mid c\land c\mid\lnot c$ & \quad \quad & $e ::= \alpha_{C}$,
\end{tabular}

\end{center}

\noindent where $p$ ranges over $\Prop$ and $\alpha_{C}$ ranges over $\Sigma^{C}$. 
\end{definition}



\begin{proposition}
\label{Proposition: Representing uniform strategies}
Let $\langle0,\lambda,s\rangle$ be a tuple. Define $H=\{\lambda'[0,k']\mid$ $\lambda'[0]\sim_{i}\lambda[0]\}$. Suppose there is a knowledge-action strategy $RS$ for agent $i$ such that

\begin{enumerate}
\item agent $i$ knows that $RS$ is complete at $\lambda[0]$, i.e.~$\langle0,\lambda,s\rangle\vDash K_{i}\Box\henceforth\bigvee_{K_i c\in\Cond(RS)} K_i c$;
\item agent $i$ knows that she is henceforth able to act according to $RS$,
i.e.~$\langle0,\lambda,s\rangle\vDash K_{i}\henceforth\Diamond[i\acc]RS$;
\item agent $i$ performs strategy $S$, i.e.~$\langle0,\lambda,s\rangle\vDash[i\perf]RS$.
\end{enumerate}

\noindent Then there is a uniform strategy $\hat{s}_{i}$ on $H$ that is play-equivalent to $s_i$.

\end{proposition}


This result is similar to Proposition \ref{Proposition: Representing positional strategies} on positional strategies. Here we see that whenever an agent performs a knowledge-action strategy of which he knows both that it is complete and that he is henceforth able to act accordingly, the agent is performing a strategy that is play-equivalent to a uniform strategy. This means that, under certain syntactically representable epistemic conditions, a syntactically characterised knowledge-action strategy corresponds with a uniform strategy in the semantic structures. 

The converse does not hold, as can be shown by providing a CEGM containing a uniform strategy that differs at two distinguishable propositionally equivalent states. But the mismatch runs deeper because of the restrictions on the conditions of knowledge-action strategies. To obtain a correspondence result in line with Corollary \ref{Corollary: Characterizing positional strategies} we believe that the language has to be extended with temporal modalities referring to the past and the conditions of knowledge-action strategies have to be modified accordingly. To keep the current exposition accessible we leave this for another occasion. In spite of these simplifications, a representation result for uniform positional strategies can be proven:

\begin{proposition}\label{Proposition: Representing uniform strategies propositionally definability}
Let $\langle0,\lambda,s\rangle$ be a profile. Define $H=\{\lambda'[0,k']\mid\lambda'[0]\sim_{i}\lambda[0]\}$. Suppose $s_{i}$ is a uniform positional strategy for agent $i$ on $H$. Suppose that $s_{i}(\St)\subseteq\Act$ is finite and that every $s_i^{-1}(\alpha_i)$ is propositionally definable.

Then there is a knowledge-action strategy $RS$ for $i$ such that 
\begin{enumerate}
\item agent $i$ knows that $RS$ is complete at $\lambda[0]$, i.e.~$\langle0,\lambda,s\rangle\vDash K_{i}\Box\henceforth\bigvee_{K_i c\in\Cond(RS)}K_i c$; 
\item agent $i$ knows that she is henceforth able to act according to $RS$,
i.e.~$\langle0,\lambda,s\rangle\vDash K_{i}\henceforth\Diamond[i\acc]RS$;
\item $\langle0,\lambda,s\rangle\vDash[i\perf]RS$.
\end{enumerate}
\end{proposition}


Note that this proposition starts with a uniform \emph{positional} strategy. This result establishes that, under certain semantic restrictions, a uniform positional strategy is represented by a knowledge-action strategy of which one knows both that it is complete and that one can henceforth act accordingly. This shows that strategic reasoning in the semantics on uniform positional strategies can be diverted to reasoning on the syntactic level about knowledge-action strategies. 

Our representation result reveals crucial underlying conditions for uniform positional strategies, which are expressible in our language. Indeed, the previous proposition shows that, under certain model restrictions, performing a uniform positional strategy implies that one is performing a knowledge-action strategy and one knows both that this knowledge-action strategy is complete and that one is henceforth able to act accordingly. Revealing such underlying conditions enhances our understanding and triggers further questions; two of such inquiries are discussed below. 

Does ensuring a property $\varphi$ by performing a uniform strategy \emph{entail} that knowing that performing this uniform strategy ensures that property? According to our representation result, this translates to questioning whether $[i\sstit]\varphi\land[i\perf]RS \land K_{i}\Box\henceforth\bigvee_{K_i c\in\Cond(RS)}K_i c\land K_{i}\henceforth\Diamond[i\acc]RS$ logically entails $K_i\Box ([i\perf]RS\rightarrow [i\sstit]\varphi)$ (where $RS$ is a knowledge-action strategy). It turns out that this indeed fails since ``in order to identify a successful strategy, the agents must consider not only the courses of action, starting from the current state of the system, but also from states that are indistinguishable from the current one.'' \cite[p.\,574]{agotnes_knowledge_2015} Uniform strategies are therefore not faithful to the expectation that ``the agent has enough control and knowledge to identify and execute a strategy that enforces [a certain property] $\varphi$.'' \cite[p.\,574]{agotnes_knowledge_2015} An agent has this control and knowledge if and only if there is a knowledge-action strategy $RS$ satisfying $K_{i}\Box\henceforth\bigvee_{K_i c\in\Cond(RS)}K_i c\land K_{i}\henceforth\Diamond[i\acc]RS\land K_i\Box ([i\perf]RS\rightarrow [i\sstit]\varphi)$. This discussion highlights the flexibility of our syntactical approach to correct the flaw of uniform strategies.

What \emph{is} a coalition's uniform strategy? Formally, it is a tuple of individuals' uniform strategies; intuitively, it is intended to capture a coalition's control and knowledge to identify and execute a strategy that enforces a certain property $\varphi$. Does a coalition's uniform strategy meet this intuition? No, it does not. In \cite[pp.~575-576]{agotnes_knowledge_2015} it is argued that ``there are several different ``modes'' in which [a coalition] can know the right strategy'', pointing to a choice between common, mutual, or distributed knowledge of the right coalitional strategy $s_C$.%
\footnote{They also mention the option that ``the strategy $s_C$ can be identified by'' (altered notation) a leader, headquarters committee, or consulting company. Our representation result suggests that the syntactical counterpart of these ``modes'' is straightforward by replacing $K_i$'s with the respective group knowledge in Proposition \ref{Proposition: Representing uniform strategies propositionally definability}. We will not pursue this suggestion in further detail here. } 
Our results, however, solicit a view, complementing the aforementioned modes, to adequately conceptualize a coalition's uniform strategies:
First, because a coalition's uniform strategy is merely a tuple of individuals' uniform strategies, a coalition's uniform strategy $s_C$ does not require that any member can \emph{identify} the coalitional strategy $s_C$. Indeed, a coalition's uniform strategy $s_C$ merely requires every member to know their part $s_i$ of the coalitional strategy $s_C$. So the \emph{object} of the members' knowledge differs. 
Second, since every member knows their part $s_i$ of a coalition's uniform strategy $s_C$, it follows that a coalition's uniform strategy entails that the coalition has distributed knowledge of the right coalitional strategy $s_C$, positioning a coalition's uniform strategy between the modes of mutual and distributed knowledge of the right coalitional strategy. 
Third, this knowledge is, however, distributed in a very particular way, namely by every member knowing \emph{their own part} $s_i$. For instance, whenever $i$ knows $j$'s part $s_j$ and $j$ knows $i$'s part $s_i$, then they have distributed knowledge of $\langle s_i,s_j\rangle$ even though it is not a uniform strategy. A coalition's uniform strategy hence does not correspond to any of the ``modes'' in \cite{agotnes_knowledge_2015}. It seems that this distinctive way of distributing knowledge solicits a  comparison, not with the proposed modes in \cite{agotnes_knowledge_2015} but, with different ways of distributing knowledge of a coalitional strategy $s_C$. 


\section{\label{Conlusion}Conclusion}
We have shown that, under certain model restrictions, a strategy that is play-equivalent to a positional strategy \emph{corresponds} to a complete proposition-action strategy. Thereby we have established a firm correspondence between a semantic strategy type and a syntactic one. 

In our research on individuals' uniform strategies, we have proven that any knowledge-action strategy of which one knows both that one can henceforth act accordingly and that it is complete \emph{represents}    a strategy that is play-equivalent to a uniform strategy. Conversely, under certain semantic restrictions, a uniform positional strategy \emph{is represented by} a knowledge-action strategy of which one knows both that one can henceforth act accordingly and that it is complete. This latter result exposes the implicit conditions of uniform positional strategies. 


The current enterprise is a crucial first step in facilitating strategic reasoning at the syntactic level. By representing several semantic strategy types and drawing novel conceptual implications we have shown the fruitfulness of our syntactic approach to enhance our understanding of semantic strategy types.

\bibliography{references}
\bibliographystyle{eptcs}

\appendix
\section{Appendix: Proof sketches of propositions}
\setcounter{proposition}{0}

\begin{proposition}
\label{Proposition: main result} Let $RS$ be a rule-based
strategy with only moment-determinate conditions. Let $s^{RS}_{C}$ be
the partial coalition strategy defined by $RS$. 
Then 
\begin{enumerate}
\item for any profile $\langle k',\lambda',s'\rangle$ such that $s^{RS}_C$ is defined on $\lambda'[0,k']$ we have that the following are equivalent: (a) $\langle k',\lambda',s'\rangle\vDash[C\acc]RS$ and (b) $s^{RS}_C(\lambda'[0,k'])=s'_C(\lambda'[0,k'])$.
\item Let $\langle k',\lambda',s'\rangle$ be a profile such that $s^{RS}_{C}$ is defined on all histories in $\outplays(\lambda'[0,k'],s'_{C})$. Then $\langle k',\lambda',s'\rangle\vDash[C\perf]RS$ if and only if $s^{RS}_{C}$ and $s_{C}'$ are play-equivalent at $\lambda'[0,k']$.
\end{enumerate}
\end{proposition}

\begin{proof}[Proof sketch:]
1. Follows from the fact that $s^{RS}_{C}$ is defined at a history iff there is exactly one way to act according to $RS$ at that history. 2. Follows straightforwardly from 1.~and property 5(b) in Definition \ref{Definition: CGM}.
\end{proof}

\begin{proposition}
\label{Proposition: Representing positional strategies}Let $s$ be a strategy profile, and let $\lambda\in\outplays(\lambda[0],s)$. Suppose there is a proposition-action strategy $RS$ for $C$ that is complete at $\lambda[0]$ such that $\langle0,\lambda,s\rangle\vDash[C\perf]RS$.
Then there is a positional strategy $\hat{s}_{C}$ that is play-equivalent to $s_{C}$ at $\lambda[0]$. 
\end{proposition}

\begin{proof}[Proof sketch:]
It is easy to show that the partial strategy defined by $RS$ is defined on and positional for histories in $\outplays(\lambda[0],s_C)$. This partial positional strategy can be trivially extended to a positional strategy, thereby proving the proposition.
\end{proof}

\begin{proposition}
\label{Proposition: Propositional definability and representing positional strategies}
Let $s_{C}$ be a positional $C$-strategy,
and let $\lambda\in\outplays(\lambda[0],s_{C})$. Suppose that $s_{C}(\St)\subseteq\Act^{C}$
is finite and that every $s_{C}^{-1}(\alpha_{C})\subseteq\St$ is
propositionally definable. Then there is a proposition-action strategy
$RS$ for $C$ that is complete at $\lambda[0]$ such that $\langle0,\lambda,s\rangle\vDash[C\perf]RS$.
\end{proposition}
\begin{proof}[Proof sketch:]
Let $\{\alpha_{C}^{1},\ldots,\alpha_{C}^{N}\}=s_{C}(\St)$, and let us denote the propositional formula defining $s_{C}^{-1}(\alpha_{C}^{n})$
by $\xi_{n}$ for each $n\leq N$. The proposition-action strategy $RS:=\{\xi_{1}\mapsto\alpha_{C}^{1},\ldots,\xi_{N}\mapsto\alpha_{C}^{N}\}$ can be used to prove the proposition.
\end{proof}

\begin{proposition}\label{Proposition: Representing uniform strategies}
Let $\langle0,\lambda,s\rangle$ be a tuple. Define $H=\{\lambda'[0,k']\mid$ $\lambda'[0]\sim_{i}\lambda[0]\}$. Suppose there is a knowledge-action strategy $RS$ for agent $i$ such that

\begin{enumerate}
\item agent $i$ knows that $RS$ is complete at $\lambda[0]$, i.e.~$\langle0,\lambda,s\rangle\vDash K_{i}\Box\henceforth\bigvee_{K_i c\in\Cond(RS)}K_i c$; 

\item agent $i$ knows that she is henceforth able to act according to $RS$,
i.e.~$\langle0,\lambda,s\rangle\vDash K_{i}\henceforth\Diamond[i\acc]RS$;

\item agent $i$ performs strategy $S$, i.e.~$\langle0,\lambda,s\rangle\vDash[i\perf]RS$.
\end{enumerate}

\noindent Then there is a uniform strategy $\hat{s}_{i}$ on $H$ that is play-equivalent to $s_i$.

\end{proposition}

\begin{proof}[Proof sketch:]
The partial strategy defined by $RS$ is defined and uniform on histories in $H$. 
\end{proof}

\begin{proposition}\label{Proposition: Representing uniform strategies propositionally definability}
Let $\langle0,\lambda,s\rangle$ be a profile. Define $H=\{\lambda'[0,k']\mid\lambda'[0]\sim_{i}\lambda[0]\}$. Suppose $s_{i}$ is a uniform positional strategy for agent $i$ on $H$. Suppose that $s_{i}(\St)\subseteq\Act$ is finite and that every $s_i^{-1}(\alpha_i)$ is propositionally definable.

Then there is a knowledge-action strategy $RS$ for $i$ such that 
\begin{enumerate}
\item agent $i$ knows that $RS$ is complete at $\lambda[0]$, i.e.~$\langle0,\lambda,s\rangle\vDash K_{i}\Box\henceforth\bigvee_{K_i c\in\Cond(RS)}K_i c$;
\item agent $i$ knows that she is henceforth able to act according to $RS$,
i.e.~$\langle0,\lambda,s\rangle\vDash K_{i}\henceforth\Diamond[i\acc]RS$;
\item $\langle0,\lambda,s\rangle\vDash[i\perf]RS$.
\end{enumerate}
\end{proposition}

\begin{proof}[Proof sketch:]
Analogous to Proposition \ref{Proposition: Propositional definability and representing positional strategies}.
\end{proof}
\end{document}